\documentclass[letterpaper, 10 pt, conference]{ieeeconf}  

\IEEEoverridecommandlockouts                              
\usepackage{float,amsmath,amsfonts,amssymb,amsthm,color,graphicx}
\usepackage{graphics} 
\usepackage{booktabs} 
\usepackage{caption, subcaption}
\usepackage{breqn}
\usepackage{bbm}
\usepackage{array, mathtools}
\usepackage{pgfplots}
\usepackage{siunitx}
\usepackage{algorithm}
\usepackage[noend]{algpseudocode}
\usepackage{tabulary}
\usepackage{stmaryrd} 
\usepackage{cite}

\makeatletter
\newsavebox{\@brx}
\newcommand{\llangle}[1][]{\savebox{\@brx}{\(\m@th{#1\langle}\)}%
  \mathopen{\copy\@brx\kern-0.5\wd\@brx\usebox{\@brx}}}
\newcommand{\rrangle}[1][]{\savebox{\@brx}{\(\m@th{#1\rangle}\)}%
  \mathclose{\copy\@brx\kern-0.5\wd\@brx\usebox{\@brx}}}
\makeatother

\usepackage{tikz}
\usepackage{pgfplots}
\usetikzlibrary{decorations.markings}
\usetikzlibrary{patterns}
\usetikzlibrary{arrows}
\usetikzlibrary{shapes}
\usetikzlibrary{fit}
 \usetikzlibrary{calc}
\usetikzlibrary{decorations.pathreplacing}
\usetikzlibrary{arrows,positioning} 
\usetikzlibrary{pgfplots.groupplots}

\usepackage{hyperref}
\hypersetup{
    colorlinks=true,
    linkcolor=black,
    filecolor=magenta,      
    urlcolor=black,
}

\usepackage{thmtools}
\declaretheorem[style=definition,name=Definition,qed=$\blacksquare$]{definition}
\declaretheorem[style=definition,name=Example,qed=$\blacksquare$]{example}

\declaretheorem[style=definition,name=Remark,qed=$\blacksquare$]{remark}

\declaretheorem[style=plain,name=Proposition,qed=$\blacksquare$]{proposition}
\declaretheorem[style=plain,name=Theorem,qed=$\blacksquare$]{theorem}

\newcommand{\R}{\mathbb{R}}

\newcommand{\eR}{\overline{\R}}
\newcommand{\X}{\mathcal{X}}

\newcommand{\Y}{\mathcal{Y}}

\newcommand{\W}{\mathcal{W}}

\newcommand{\rect}[1]{\llbracket #1 \rrbracket}

\title{\LARGE {\bf Improving the Fidelity of Mixed-Monotone Reachable Set Approximations via State Transformations}}

\author{Matthew Abate and Samuel Coogan
\thanks{ This work was supported in part by the National Science Foundation under grant \#1749357 and the Air Force Office of Scientific Research under award FA9550-19-1-0015.}
\thanks{M. Abate is with the School of Mechanical Engineering and the School of Electrical and Computer Engineering, Georgia Institute of Technology, Atlanta, 30332, USA {\tt\small Matt.Abate@GaTech.edu}.}
\thanks{S. Coogan is with the School of Electrical and Computer Engineering and the School of Civil and Environmental Engineering, Georgia Institute of Technology, Atlanta, 30332, USA {\tt\small Sam.Coogan@GaTech.edu}.}
}

\begin{document}

\maketitle
\thispagestyle{empty}
\pagestyle{empty}

\begin{abstract}
Mixed-monotone systems are separable via a decomposition function into increasing and decreasing components, and this decomposition function allows for embedding the system dynamics in a higher-order monotone embedding system.  Embedding the system dynamics in this way facilitates the efficient over-approximation of reachable sets with hyperrectangles, however, unlike the monotonicity property, which can be applied to compute, \emph{e.g.}, the \emph{tightest} hyperrectangle containing a reachable set, the application of the mixed-monotonicity property generally results in conservative reachable set approximations.
In this work, explore conservatism in the method and we consider, in particular, embedding systems that are monotone with respect to an alternative partial order.  This alternate embedding system is constructed with a decomposition function for a related system, formed via a linear transformation of the initial state-space. We show how these alternate embedding systems allow for computing reachable sets with improved fidelity, \emph{i.e.}, reduced conservatism.
\end{abstract}

\section{Introduction}
A dynamical system is mixed-monotone if there exists a related \textit{decomposition function} that decomposes the system's vector field into increasing and decreasing components; mixed-monotonicity applies to continuous-time systems \cite{Gouze:1994qy, Coogan:2016, Nonmonotone, SuffMM}, discrete-time systems \cite{Smith2006}, as well as systems with disturbances \cite{TIRA, LTLandMM, coogan2015efficient}, and it generalizes the \emph{monotonicity} property of dynamical systems for which trajectories maintain a partial order over states \cite{smith2008monotone, monotonicity}.

For an $n$-dimensional mixed-monotone system with a disturbance input, it is possible to construct a $2n$-dimensional \emph{embedding system} from the decomposition function.  This embedding system contains no disturbances and it is monotone with respect to a particular partial order.  Thus, tools from monotone systems theory can be applied to the embedding system to conclude properties of the original dynamics; in particular, such approaches are useful to efficiently approximate reachable sets using hyperrectangles. For example, it is shown in \cite{coogan2015efficient, TIRA, LTLandMM} how finite-time forward reachable sets for the original system are efficiently approximated via a single simulation of the embedding system, and this procedure is extended in \cite{Invariance4MM} for the approximation of backward-time reachable sets.  These works assume a hyperrectangular initial set of interest, and the approximations derived from their procedures are also hyperrectangles.

Unlike the monotonicity property, which can be applied to compute, \emph{e.g.}, the \emph{tightest} hyperrectangle containing a reachable set \cite{smith2008monotone}, the application of the mixed-monotonicity property is known to generally result in conservative reachable set approximations \cite{coogan2015efficient, TIRA, LTLandMM, Invariance4MM}.  In this work, we explore two main ways of reducing the conservatism in the approximation of reachable sets: (i) using alternative and/or multiple decomposition functions, and (ii) using alternative and/or multiple partial orders.

The first topic was recently explored in continuous-time \cite{abate2020tight} and it is now known that all systems are mixed-monotone with a unique \emph{tight} decomposition function that computes reachable sets with less conservatism than any other decomposition function. This tight construction is defined as an optimization problem and may not always be practically computable.
Thus, in some instances, employing a different decomposition function construction may be preferable; see  \cite{7799445, SuffMM, TIRA, LTLandMM} for an algorithm to generate decomposition functions for systems with uniformly bounded Jacobian matrices, and see also \cite{Invariance4MM} for an algorithm to generate decomposition functions for systems defined by polynomial vector fields.  Our first result is to show how two initial decomposition functions for a given system can be combined in a piecewise fashion to create a new decomposition function for the same system that approximates reachable sets with accuracy at least as good as employing both initial decomposition functions independently and forming a reachable set approximation as the intersection of the approximation derived from each function.  
This method for reducing conservatism is particularly useful when both initial decomposition functions are derived using the Jacobian bound approach from \cite{7799445, SuffMM, TIRA, LTLandMM}; this approach can produce multiple distinct decomposition functions for the same system and combining these functions allows for added fidelity.

The main results of this paper, however, deal with the second topic regarding alternative partial orders. In particular, we consider the standard componentwise partial orders in a linearly transformed state-space, and we observe that inequality intervals in the transformed space correspond to parallelotopes in the original state-space. Thus, it is possible to compute parallelotope over-approximations of reachable sets by applying the standard mixed-monotonicity tools with the new order. We present two methods for reducing conservatism in this manner: (i) several different partial orders can be used so that the reachable set of the system is known to lie in the intersection of the approximation derived from each partial order, (ii) in certain cases, a linear transformation can be found to transform the system to a monotone system.  Moreover, as a tight decomposition function is known to exist for any given transformed system, there exists an analogous notion of tightness with respect to any given parallelotope shape. 

The results and tools created in this work are demonstrated through three examples and a case study\footnote{The code that accompanies these examples and generates the figures in this work is publicly available through the GeorgiaTech FactsLab GitHub: \url{https://github.com/gtfactslab/Abate\_ACC2021\_2}.}.

\section{Notation}

Let $(x,\, y)$ denote the vector concatenation of $x,\, y \in \R^n$, \emph{i.e.}, $(x,\,y) := [x^T \, y^T]^T \in \R^{2n}$, and let $\preceq$ denote the componentwise vector order, \emph{i.e.}, $x\preceq y$ if and only if $x_i \leq y_i$ for all $i\in\{1,\cdots, n\}$ where vector components are indexed via subscript. We say that $x,\, y \in \R^n$ are ordered when either $x \preceq y$ or $y \preceq x$.

Given $x, y\in\mathbb{R}^n$ with $x\preceq y$, 
\begin{equation*}
[x,\,y]:=\left\{z\in \R^n \,\mid\, x \preceq z \text{ and } z \preceq y\right \}
\end{equation*}
denotes the hyperrectangle defined by the endpoints $x$ and $y$, and given a nonsingular transformation matrix $T\in \R^{n\times n}$,
\begin{equation*}
[x,\,y]_T :=\left\{z\in \R^n \,\mid\, T^{-1} z \in [x,\, y] \right \}
\end{equation*}
denotes the parallelotope defined by the endpoints $x$ and $y$ and shape matrix $T$. 
Given $a=(x,\,y) \in \R^{2n}$ with $x \preceq y$, we denote by $\rect{a}$ the hyperrectangle formed by the first and last $n$ components of $a$, \emph{i.e.},  $\rect{a}:=[x,\,y]$, and likewise $\rect{a}_T:=[x,\,y]_T$.

Let $\preceq_{\rm SE}$ denote the \emph{southeast order} on $\eR^{2n}$ defined by
\begin{equation*}
(x,\, x') \preceq_{\rm SE} (y,\, y')
 \:\: \Leftrightarrow  \:\:  x \preceq y\text{ and } y' \preceq x'
\end{equation*}
where $x,\, y,\, x',\, y' \in \eR^n$.  In the case that $x \preceq x'$ and $y \preceq y'$, observe that
\begin{equation}
\label{eq:order_to_box}
(x,\, x') \preceq_{\rm SE} (y,\, y')
 \:\: \Leftrightarrow  \:\:
[\,y,\, y'\,] \subseteq [\,x,\, x'\,].
\end{equation}

\section{Preliminaries}
We consider a dynamical system with disturbances
\begin{equation}\label{sys}
    \dot{x} = F(x,\, w)
\end{equation}
with state $x \in \X \subseteq \R^n$ and disturbance input $w \in \W \subset \R^m$, where $\W = [\underline{w},\, \overline{w}]$ for some $\underline{w} \preceq \overline{w}$.

Let $\Phi(t;\, x,\, \mathbf{w}) \in \X$ denote the unique state of \eqref{sys} reached at time $t$ when starting from state $x$ at time $0$ and evolving subject to the piecewise continuous signal $\mathbf{w} : [0,\, t] \rightarrow \W$.  We allow for finite-time escape so that $\Phi(t;\, x,\, \mathbf{w})$ need not exist for all $t$, however, $\Phi(t;\, x,\, \mathbf{w})$ is understood to exist only when $\Phi(\tau;\, x,\, \mathbf{w}) \in \X$ for all $\tau \in [0,\, t]$, and statements involving $\Phi(t;\, x,\, \mathbf{w})$ are understood to apply only when $\Phi(t;\, x,\, \mathbf{w})$ exists.
For given $\X_0 \subseteq \X$ and $t \geq 0$, we denote by $R(t;\, \X_0)$ the time-$t$ forward reachable set of \eqref{sys} from $\X_0$, that is,
\begin{multline}\label{forward_reach}
    R(t;\, \X_0) = \{\Phi(t;\, x,\, \mathbf{w})\in \X \:\vert\: 
    x \in \X_0,\\ \mathbf{w} : [0,\, t] \rightarrow \W\}.
\end{multline}

We begin by recalling fundamental results in mixed-monotone systems theory. 

\begin{definition}[Mixed-Monotonicity]\label{def1}
Given a locally Lipschitz continuous function $d : \X \times \W \times \X \times \W \rightarrow \R^n$, the system \eqref{sys} is \textit{mixed-monotone with respect to $d$}  if
\begin{enumerate}
    \item For all $x \in \X$ and all $w \in \W$, $d(x,\,w,\, x,\, w) = F(x,\, w)$ holds.
    \item For all $i,\, j \in \{1,\, \cdots,\, n\}$, with $i \neq j$, $\frac{\partial d_i}{\partial x_j}(x, \,w,\, \widehat{x},\, \widehat{w}) \geq 0$ holds for all ordered $(x,\,w),\, (\widehat{x},\,\widehat{w}) \in \X\times \W$ such that $\frac{\partial d}{\partial x}$ exists.
    \item For all $i,\, j \in \{1,\, \cdots,\, n\}$, $\frac{\partial d_i}{\partial  \widehat{x}_j}(x, \,w,\, \widehat{x},\, \widehat{w}) \leq 0$
    holds for all ordered $(x,\,w),\, (\widehat{x},\,\widehat{w}) \in \X\times \W$ such that $\frac{\partial d}{\partial \widehat{x}}$ exists.
    \item For all $i\in \{1,\, \cdots,\, n\}$ and all $j\in \{1,\, \cdots,\, m\}$, $\frac{\partial d_i}{\partial  w_j}(x, \,w,\, \widehat{x},\, \widehat{w}) 
    \geq 0 \geq 
    \frac{\partial d_i}{\partial  \widehat{w}_j}(x, \,w,\, \widehat{x},\, \widehat{w})$
    holds for all ordered $(x,\,w),\, (\widehat{x},\,\widehat{w}) \in \X\times \W$ such that $\frac{\partial d}{\partial w}$ and $\frac{\partial d}{\partial \widehat{w}}$ exist.
    \qedhere
\end{enumerate}
\end{definition}

When \eqref{sys} is mixed-monotone with respect to $d$, $d$ is said to be a decomposition function for \eqref{sys}.  Given $d$, the system 
\begin{equation}\label{embedding}
    \begin{bmatrix}
    \dot{x} \\ \dot{\widehat{x}}
    \end{bmatrix}
    = E(x,\, \widehat{x}) = 
    \begin{bmatrix}
    d(x,\, \underline{w},\, \widehat{x},\, \overline{w}) \\
    d(\widehat{x},\, \overline{w},\, x,\, \underline{w})
    \end{bmatrix}
\end{equation}
is the \emph{embedding system relative to $d$}, and $E$ is the \emph{embedding function relative to $d$}. We let $\Phi^{E}(t;\, a)$ denote the unique state of \eqref{embedding} reached at time $t \geq 0 $ when beginning from state $a \in \X \times \X$ at time $0$.

We show in the following Proposition how approximations of reachable sets for \eqref{sys} are efficiently computed via a single simulation of the embedding system \eqref{embedding}. 

\begin{proposition}\label{prop:p1} 
Let \eqref{sys} be mixed-monotone with respect to $d$, and let $\X_0 = [\underline{x},\, \overline{x}]$ for some $\underline{x},\, \overline{x}\in \X$ with $\underline{x} \preceq \overline{x}$. Then $R(t;\, \X_0) \subseteq \rect{\Phi^{E}( t;\, (\underline{x},\, \overline{x}))}.$
\end{proposition}

The proof of Proposition \ref{prop:p1} appears in \cite[Appendix B1]{TIRA} and in an extended version of \cite{Invariance4MM}.

The application of Proposition \ref{prop:p1} is known to provide conservative estimates of reachable sets, and it is natural to wonder whether fidelity can be improved. In the following section, we discuss the three main ways that conservatism enters the approach, and in the later sections we study methods for reducing this approximation conservatism. 

The mixed-monotonicity property generalises the monotonicity of dynamical systems, for which trajectories maintain a partial order over states.

\begin{definition}\label{def:monotone}
The system \eqref{sys} is a \emph{monotone} dynamical system if
\begin{enumerate}
    \item For all $i,\, j \in \{1,\, \cdots,\, n\}$, with $i \neq j$, $\tfrac{\partial F_i}{\partial x_j}(x,\, w) \geq 0$ holds for all $x \in \X$ and all $w \in \W$.
    \item For all $i\in \{1,\, \cdots,\, n\}$ and all $j\in \{1,\, \cdots,\, m\},$ $\tfrac{\partial F_i}{\partial w_j}(x,\, w) \geq 0$ holds for all $x \in \X$ and all $w \in \W$. \qedhere
\end{enumerate}
\end{definition}

When \eqref{sys} is monotone, \eqref{sys} is mixed-monotone with respect to $d(x,\, w,\, \widehat{x},\, \widehat{w}) = F(x,\, w)$, and this decomposition function yields the tightest hyperrectangle containing $R(t;\, \X_0)$ when used with Proposition \ref{prop:p1}; that is, $\rect{\Phi^{E}( t;\, (\underline{x},\, \overline{x}))}$ contains $R(t;\, [\underline{x},\, \overline{x}])$ and no proper hyperrectangular subset of $\rect{\Phi^{E}( t;\, (\underline{x},\, \overline{x}))}$ contains $R(t;\, [\underline{x},\, \overline{x}])$.

\section{Discussion on Conservatism in the Method}
As discussed previously, the application of Proposition \ref{prop:p1} is known to generally result in conservative reachable set approximations, and  conservatism enters the method in three main ways: (a) when a \emph{non-tight} decomposition function is employed, (b) when the decomposition function $d$ varies quickly from the vector field $F$, and (c) when the initial set $\X_0$ is poorly approximated by a hyperrectangle.

Generally, a mixed-monotone system, as in \eqref{sys}, will be mixed-monotone with respect to many decomposition functions, although certain decomposition functions will provide tighter approximations of reachable sets than others when used with Proposition \ref{prop:p1}.  Thus, given a mixed-monotone system, and perhaps several decomposition functions for that system, it is natural to wonder which decomposition function is preferable in application. We use the term \emph{type-(a) conservatism} to denote the approximation conservatism added when a poor decomposition function is employed with Proposition \ref{prop:p1}. Type-(a) conservatism was recently explored in \cite{abate2020tight} and we provide additional analysis in Section \ref{Sec5}.  In particular, we recall that every mixed-monotone system induces a unique \emph{tight decomposition function} that provides a tighter approximation of reachable sets than any other decomposition function for \eqref{sys} when used with Proposition \ref{prop:p1}.  This decomposition function is defined as an optimization problem and thus may not always be practically computable.  To that end, we show additionally how several, perhaps non-tight, decomposition functions for \eqref{sys} can be {combined} to form a new decomposition function for \eqref{sys} that, when used with Proposition \ref{prop:p1}, provides tighter approximations of reachable sets than are attainable by employing each initial decomposition function separately with Proposition \ref{prop:p1} and then forming a reachable set approximation as the intersection of the approximations derived from each.

Even when a tight decomposition function is used with Proposition \ref{prop:p1}, the derived reachable set approximation may still be overly conservative.  Specifically, employing a tight decomposition function does not guarantee that no proper hyperrectangular subset of $\rect{\Phi^{E}( t;\, (\underline{x},\, \overline{x}))}$ contains $R(t;\, [\underline{x},\, \overline{x}])$, in contrast to the case when the dynamics are monotone.  This form of conservatism occurs when the decomposition function $d$ varies quickly from the vector field $F$, and we hereafter refer to this form of conservatism as \emph{type-(b) conservatism}.  In Section \ref{Sec7}, we address type-(b) conservatism, and we show how it is mitigated by considering alternate partial orders on $\X$; that is, we show that \eqref{sys} may be monotone with respect to a different partial order than that considered in Definition \ref{def:monotone} and, in this case, a parallelotope approximation of $R(t;\, \X_0)$ can be derived such that no proper parallelotope subset of this approximation contains $R(t;\, \X_0)$.  

Lastly, conservatism can enter the method when $\X_0$ is poorly approximated by a hyperrectangle.  While the hypothesis of Proposition \ref{prop:p1} assumes a hyperrectangular set of interest $\X_0$, the basic procedure holds for different set geometries by over-approximating the initial set with a hyperrectangle; in particular, if $\X_0 \subset [\underline{x}, \overline{x}]$ for some $\underline{x} \preceq \overline{x}$, then $R(t;\, \X_0) \subseteq \rect{\Phi^{E}( t;\, (\underline{x},\, \overline{x}))}$.  However, if $[\underline{x}, \overline{x}]$ poorly approximates $\X_0$, then $\rect{\Phi^{E}( t;\, (\underline{x},\, \overline{x}))}$ will poorly approximate $R(t;\, \X_0)$, and this approximation conservatism is refereed to as \emph{type-(c) conservatism}. In Sections \ref{Sec7} and \ref{Sec8}, we show how alternative partial orders on $\X$, as in those discussed previously, allow for ways of reducing type-(c) conservatism.

It is important to note that reachable set approximations derived from Proposition \ref{prop:p1} may be conservative, even when types-(a), (b) and (c) conservatism are absent.  That is, even when a tight decomposition function is used, the system \eqref{sys} is monotone, and $\X_0$ is hyperrectangular, one will generally find that $R(t;\, \X_0) \neq \rect{\Phi^{E}( t;\, (\underline{x},\, \overline{x}))}$.  This approximation conservatism, referred to hereafter as \emph{type-(d) conservatism}, is inherent in Proposition \ref{prop:p1} and cannot be mitigated using the theory discussed thus far.  We address type-(d) conservatism in Section \ref{Sec7}; we observe, in particular,  that $R(t;\, \X_0)$ is constrained to the intersection of several independent approximations derived from related systems to \eqref{embedding}, and we show through example how forming reachable set approximations in this way mitigates type-(d) conservatism.

We summarise the proceeding discussion in the following remark. 
\begin{remark}
Four main forms of conservatism arise in the application of Proposition \ref{prop:p1}:
\begin{itemize}
    \item[(a)] Type-(a) conservatism occurs when a non-tight decomposition function is used. 
    \item[(b)] Type-(b) conservatism occurs when the decomposition function $d$ varies quickly from the vector field $F$.
    \item[(c)] Type-(c) conservatism occurs when the initial set $\X_0$ is poorly approximated by a hyperrectangle.
    \item[(d)] Type-(d) conservatism is inherent to Proposition \ref{prop:p1}, and cannot be mitigated using the theory discussed thus far.
\end{itemize}
Depending on the structure of one's specific system, decomposition function, and initial set, these forms of conservatism can each occur independently of one another.  In the following Sections we study each form of conservatism, and show how approximations of reachable sets for nonlinear systems can be improved using the theory of mixed-monotonicity.
\end{remark}

\section{Reducing Conservatism via Decomposition Function Analysis}\label{Sec5}
Addressing type-(a) conservatism caused by a poor choice of decomposition function for \eqref{sys} requires constructing an alternative decomposition function for the same system. This issue was recently explored in \cite{abate2020tight} where it is shown that all systems of the form \eqref{sys} are mixed-monotone with respect to a decomposition function $d$ defined element-wise by
\begin{multline}\label{opt_decomp}
    d_i(x,\, w,\, \widehat{x},\, \widehat{w}) = \\
    \begin{cases}
    \min\limits_{
    \substack{
        y \in  [x,\, \widehat{x}]\\
        y_i = x_i \\
        z\in [w,\, \widehat{w}]
    }
    }
    F_i(y,\, z) & \text{if $(x,\,w) \preceq (\widehat{x},\, \widehat{w})$} \\
    \vspace{-.3cm}
    \\
    \max\limits_{
    \substack{ y \in  [\widehat{x},\, x]\\
               y_i = x_i \\
               z \in [\widehat{w},\, w]
               }
    } 
    F_i(y,\, z) & \text{if $(\widehat{x},\, \widehat{w}) \preceq (x,\, w)$}.
    \end{cases}
\end{multline}
We refer to $d$ constructed  in \eqref{opt_decomp} as the tight decomposition function for \eqref{sys} and, importantly, $d$ provides a tighter approximation of reachable sets than any other decomposition function for \eqref{sys} when used with Propositions \ref{prop:p1} \cite{abate2020tight}.  Thus, applying Proposition \ref{prop:p1} with \eqref{opt_decomp} ensures that the procedure does not suffer from type-(a) conservatism. See also \cite{ozaytight} for a discrete time analogue of \eqref{opt_decomp}.

The paper \cite{abate2020tight} shows how, in certain instances, a tight decomposition function for \eqref{sys} is attainable in closed form.  However, generally, the application of \eqref{opt_decomp} is prevented by its construction as an optimization problem. For this reason, computing alternative decomposition functions for \eqref{sys} via other means can be useful; see  \cite{7799445, SuffMM, TIRA, LTLandMM} for an algorithm to generate decomposition functions for systems with uniformly bounded Jacobian matrices, and see also \cite{Invariance4MM} for an algorithm to generate decomposition functions for systems defined by polynomial vector fields.  These algorithms, however, have no tightness guarantees when \eqref{sys} is not monotone.

Our first result is to show how two initial, perhaps non-tight, decomposition functions for a given system can be combined in a piecewise fashion to create a new decomposition function for the same system that approximates reachable sets with greater accuracy than either of its components.

\begin{proposition}\label{prop:p2}
Let \eqref{sys} be mixed-monotone with respect to both $d^1$ and $d^2$.  Then \eqref{sys} is mixed-monotone with respect to $d$ defined element-wise by
\begin{multline}\label{comb_decomp}
    d_i(x,\, w,\, \widehat{x},\, \widehat{w}) = \\
    \begin{cases}
    \max\{d_i^1(x,\, w,\, \widehat{x},\, \widehat{w}), d_i^2(x,\, w,\, \widehat{x},\, \widehat{w})\} \\ 
    \hspace{4cm} \text{if $(x,\, w) \preceq (\widehat{x},\, \widehat{w}),$ } \\ 
    \min\{d_i^1(x,\, w,\, \widehat{x},\, \widehat{w}), d_i^2(x,\, w,\, \widehat{x},\, \widehat{w})\} \\ 
    \hspace{4cm} \text{if $(\widehat{x},\, \widehat{w}) \preceq(x,\, w).$ } \\ 
    \end{cases}
\end{multline}
Moreover, denoting by $E,\, E^1,\, E^2$ the embedding functions relative to $d,\, d^1,\, d^2$, respectively, we have that 
\begin{equation}\label{intersectioneq}
    \rect{\Phi^E(t;\, (\underline{x},\, \overline{x}))} \subseteq 
    \rect{\Phi^{E^1}(t;\, (\underline{x},\, \overline{x}))}
    \cap
    \rect{\Phi^{E^2}(t;\, (\underline{x},\, \overline{x}))}
\end{equation}
for all $t \geq 0$ and all $\underline{x}\preceq\overline{x}$.
\end{proposition}
\begin{proof}
We first show that $d$ from \eqref{comb_decomp} is a decomposition function for \eqref{sys}. Since $d^1$ and $d^2$ are decomposition functions for \eqref{sys}, 
$d_i^1(x,\, w,\, x,\, w) = d_i^2(x,\, w,\, x,\, w) = F_i(x,\, w)$ for all 
$i\in \{1,\, \cdots,\,n\}$, all $x\in \X$, and all $w \in \W$, and therefore $d$ satisfies the first condition in Definition \ref{def1}. Note also that $d^1$ and $d^2$ are both increasing in their first two arguments and decreasing in their second two arguments; thus $d$ satisfies the conditions 2, 3 and 4 from Definition \ref{def1}. Therefore, the system \eqref{sys} is mixed-monotone with respect to $d$.

We show that \eqref{intersectioneq} holds, by showing that $\rect{\Phi^E(t;\, (\underline{x},\, \overline{x}))} \subseteq \rect{\Phi^{E^1}(t;\, (\underline{x},\, \overline{x}))}$ for all $t\geq 0,$ and $\rect{\Phi^E(t;\, (\underline{x},\, \overline{x}))} \subseteq \rect{\Phi^{E^2}(t;\, (\underline{x},\, \overline{x}))}$ follows from a reflexive argument. The construction \eqref{comb_decomp} implies 
\begin{equation}
    d(x,\, w,\, \widehat{x},\, \widehat{w}) \succeq d^1(x,\, w,\, \widehat{x},\, \widehat{w})
\end{equation}
when $(x,\, w) \preceq (\widehat{x},\, \widehat{w}),$ and 
\begin{equation}
    d(x,\, w,\, \widehat{x},\, \widehat{w}) \preceq d^1(x,\, w,\, \widehat{x},\, \widehat{w})
\end{equation}
when $(\widehat{x},\, \widehat{w}) \preceq(x,\, w).$ Thus, for all $\underline{x} \preceq \overline{x}$,
\begin{equation}
    E^1(\underline{x},\, \overline{x}) \preceq_{\rm SE} E(\underline{x},\, \overline{x}),
\end{equation}
and therefore 
\begin{equation}
    \Phi^{E^1}(t;\, (\underline{x},\, \overline{x})) \preceq_{\rm SE} \Phi^E(t;\, (\underline{x},\, \overline{x}))
\end{equation}
holds for all $t \geq 0$.  For any embedding function, as in $E^1,\, E$, the space $\{(x,\, \widehat{x}) \in \R^n\times R^n \, \vert\, x \preceq \widehat{x}\}$ is forward invariant for \eqref{embedding} \cite{Invariance4MM}, and therefore we have $\rect{\Phi^E(t;\, (\underline{x},\, \overline{x}))} \subseteq \rect{\Phi^{E^1}(t;\, (\underline{x},\, \overline{x}))}.$ Therefore,  $E^1,$ $E^2$, $E$ satisfy \eqref{intersectioneq}.  This completes the proof.
\end{proof}

Proposition \ref{prop:p2} shows how multiple decomposition functions for \eqref{sys} are combined to construct a new decomposition function for \eqref{sys}; this new decomposition function, when used with Proposition \ref{prop:p1}, provides tighter approximations of reachable sets than are attainable by, in particular, employing both initial decomposition functions and forming a reachable set approximation as the intersection of the approximation derived from each function.  Thus, employing $d$ from \eqref{comb_decomp} reduces type-(a) conservatism in the method.  Note that when either $d^1$ or $d^2$ is a tight decomposition function for \eqref{sys}, $d$ from \eqref{comb_decomp} will always resolve to the tight decomposition function \eqref{opt_decomp}. Thus the application of Proposition \ref{prop:p2} is beneficial only when non-tight decomposition functions for \eqref{sys} are known, and this fact is intuitive as only non-tight decomposition functions are subject to type-(a) conservatism.
Moreover, this method for reducing conservatism is particularly useful, when, say, $d^1$ and $d^2$ are both derived using the Jacobian bound approach appearing in \cite{Invariance4MM, TIRA}; applying this approach can produce multiple distinct decomposition functions for the same system and these decomposition functions can be combined using Proposition \ref{prop:p2} to allow for added fidelity.

While employing $d$ from \eqref{comb_decomp} reduces type-(a) conservatism in the method, this approach is still susceptible to types-(b), (c) and (d) conservatism.
In Section \ref{Sec7}, we take a different approach, and show how types-(b), (c) and (d) conservatism can be reduced when multiple partial orders are considered on $\X$.

\section{Applying the Tools of Mixed-Monotonicity with Alternate Partial Orders}\label{Sec6}
In this section, we show how the tools of mixed-monotonicity, which are traditionally employed with the standard order $\preceq$,  extend to alternate partial orders in a similar way.

Consider the state transformation of \eqref{sys} formed by taking a linear transformation on the state-space
\begin{equation}\label{transform}
    y = T^{-1}x
\end{equation}
where $x\in \X$ is the state of \eqref{sys} and where $T \in \R^{n\times n}$ is a nonsingular transformation matrix.  Under the transformation \eqref{transform}, the transformed dynamics of $y$ become 
\begin{equation}\label{trans_sys}
    \dot{y} = T^{-1}F(Ty,\, w) := F_{T}(y,\, w)
\end{equation}
with state $y \in \Y = \{T^{-1}x \;\vert\; x \in \X\}$ and disturbance input $w \in \W$. Further, the systems \eqref{sys} and \eqref{trans_sys} are related in the following way: for all $x \in \X$, all $t\geq 0$ and all piecewise continuous $\mathbf{w}: [0,\, t] \rightarrow \W$, we have $\Phi(t;\, x,\, \mathbf{w}) = T \Psi(t;\, T^{-1}x,\, \mathbf{w})$, where $\Psi$ denotes the state transition function for \eqref{trans_sys}.

We show next how a decomposition function for \eqref{trans_sys} enables the approximation of forward reachable sets for \eqref{sys}.

\begin{theorem}\label{thrm1} 
For some nonsingular $T \in \R^{n \times n}$, let \eqref{trans_sys} be mixed-monotone with respect to $d$ and let $\X_0 = [\underline{y},\, \overline{y}]_T\subseteq \mathcal{X}$ for some $\underline{y} \preceq \overline{y}$. Then $R(t;\, \X_0) \subseteq \rect{\Phi^{E}( t;\, (\underline{y},\, \overline{y}))}_T$, where $R(t;\, \X_0)$ denotes the reachable set of the original dynamics \eqref{sys} as defined in \eqref{forward_reach} and $\Phi^E$ denotes the flow of the embedding system constructed from $d$ as defined in \eqref{embedding}.
\end{theorem}
\begin{proof}
Choose $x \in \X_0$, then $T^{-1}x \in  [\underline{y},\, \overline{y}]$.  For all $\mathbf{w}$, Proposition \ref{prop:p1} implies that $\Psi(t;\, y,\, \mathbf{w}) \in \rect{\Phi^{E}( t;\, (\underline{y},\, \overline{y}))}$ for all $y \in [\underline{y},\, \overline{y}]$ and all $t \geq 0$, and therefore $\Psi(t;\, T^{-1}x,\, \mathbf{w}) \in \rect{\Phi^{E}( t;\, (\underline{y},\, \overline{y}))}$ holds.  Moreover, $\Phi(t;\, x,\, \mathbf{w}) = T \Psi(t;\, T^{-1}x,\, \mathbf{w}) \in \rect{\Phi^{E}( t;\, (\underline{y},\, \overline{y}))}_T$.
\end{proof}

Theorem \ref{thrm1} extends the applicability Proposition \ref{prop:p1} to the case of parallelotope initial sets $\X_0$, and the approximations derived from the application of Theorem \ref{thrm1} will also be parallelotopes.  As such, the results of Theorem \ref{thrm1} subsume those of Proposition \ref{prop:p1} as a special case by taking $T = I_n$ where $I_n$ is the $n\times n$ identity matrix.
We demonstrate the application of Theorem \ref{thrm1} in the following example.

\begin{example}\label{Example1}
Consider the system
\begin{equation}\label{EX1_sys}
\begin{bmatrix}
\dot{x}_1\\
\dot{x}_2
\end{bmatrix}
=
F(x,\, w)
=
\begin{bmatrix}
x_1 x_2 + w\\
x_1 + 1
\end{bmatrix}
\end{equation}
with state-space $\X = \R^2$ and disturbance-space $\W = [0,\, 1/4]$.  We assume a parallelotope set of initial conditions $\X_0 = [\underline{y},\, \overline{y}]_T$ for 
\begin{equation}\label{EX1_params}
    \underline{y} = 
    \begin{bmatrix}
    0 \\ -1/4
    \end{bmatrix}, \quad
    \overline{y} = 
    \begin{bmatrix}
    1/4 \\ 0
    \end{bmatrix}, \quad
    T = \begin{bmatrix}
    1 & -2 \\
    1 & 1
    \end{bmatrix}
\end{equation}
and we aim to approximate $R(1;\, \X_0)$ using Theorem \ref{thrm1}.

A decomposition function for \eqref{trans_sys} is formed using the tight construction \eqref{opt_decomp}, and its embedding system is simulated forward in time in order to approximate $R(1,\, \X_0)$.  We solve \eqref{opt_decomp} at each timestep of the simulation using \textsc{fminbnd.m}, a MATLAB optimization tool. We show $\X_0$ graphically in Figure \ref{fig1}, along with $R(1;\, \X_0)$ and its respective approximation as derived in Theorem \ref{thrm1}. 

Note that the approximations derived thus far do not suffer from types-(a) and (c) conservatism; this is due to the fact that $\X_0$ is parallelotopic and we employ a tight decomposition function in the procedure. However, we find that $R(1;\, \X_0) \neq \rect{\Phi^{E}(1;\, [\underline{y},\, \overline{y}])}_T$ and this is due to types-(b) and (d) conservatism.

Additionally, note that the aforementioned procedure for computing parallelotope approximations of forward reachable sets can be extended to approximate backward reachable sets in a similar way.  In particular, in \cite{Invariance4MM}, it is shown how a decomposition function for $\dot{x} = -F(x,\, w)$ is used to compute a hyperrectangular over-approximation of
\begin{multline}\label{backward_reach}
    S(t;\, \X_0) = \{x \in \X \:\vert\: \Phi(t;\, x,\, \mathbf{w}) \in \X_0\\ \text{ for some } \mathbf{w} : [0,\, t] \rightarrow \W\},
\end{multline}
and we observe the same technique can be employed with $\dot{y} = - F_T(y,\, w)$.  An example is shown in Figure \ref{fig1} where we compute a parallelogram approximation of $S(1;\, \X_0)$ using a tight decomposition function for $\dot{y} = -F_T(y,\, w)$. 
\end{example}

\begin{figure}[t]
    \hspace{.6cm}
    \input{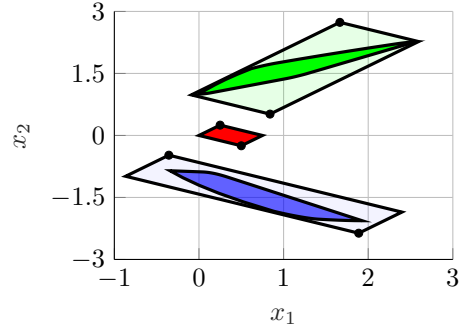}
    \caption{Example \ref{Example1}. $\X_0$ is shown in red.  $R(1;\, \X_0)$ is shown in green, with a parallelogram over-approximation shown in light green. $S(1;\, \X_0)$ is shown in blue, with a parallelogram over-approximation shown in light blue. Note that $R(1;\, \X_0) \neq \rect{\Phi^{E}(1;\, [\underline{y},\, \overline{y}])}_T$ and this is due to types-(b) and (d) conservatism in the method.
    }
    \label{fig1}
\end{figure}

Note that Theorem \ref{thrm1} induces an analogous notion of conservatism to that of Proposition \ref{prop:p1}.  That is, one may not have access to a tight decomposition function for \eqref{trans_sys}, and in this case the application of Theorem \ref{thrm1} is subject to type-(a) conservatism. In addition, when $T$ is chosen poorly, the system \eqref{trans_sys} may only induce decomposition functions which vary quickly from $F_y$,  and in this case the application of Theorem \ref{thrm1} is subject to type-(b) conservatism. Last, the set of interest $\X_0$ may be poorly approximated by a parallelotope $[\underline{y}, \overline{y}]_T$, and in this case the application of Theorem \ref{thrm1} is subject to type-(c) conservatism.  Note however, that type-(c) conservatism is always mitigated when $\X_0$ is a singleton set.

\begin{figure*}[t!]
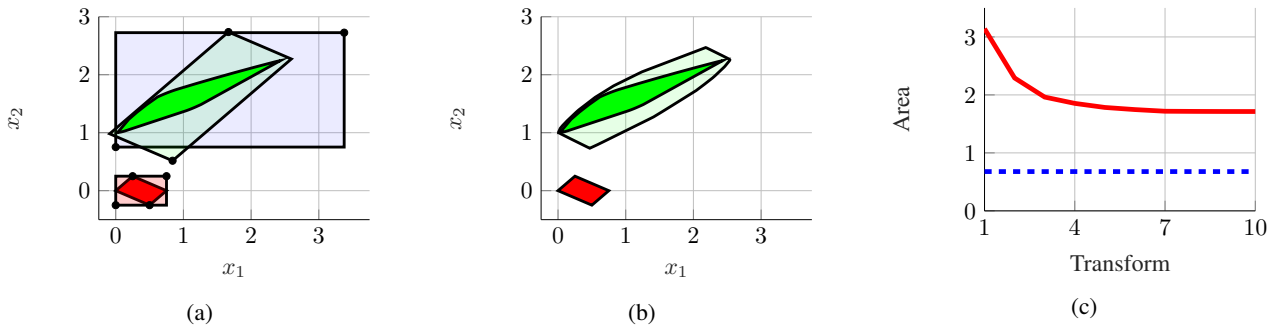

    \begin{subfigure}{0.31\textwidth}
    \scalebox{.9}{
        \input{F2a.tikz}
        }
        \caption{ }
        \label{fig2a}
    \end{subfigure}
    ~
    \begin{subfigure}{0.31\textwidth}
    \scalebox{.9}{
        \input{F2b.tikz}
        }
        \caption{ }
        \label{fig2b}
    \end{subfigure}
    ~
    \begin{subfigure}{0.31\textwidth}
    \scalebox{.9}{
%
%
\begin{tikzpicture}

\begin{axis}[%
width=4cm,  
height=3cm, 
at={(0cm,0cm)},
scale only axis,
xmin=1,
xmax=10,
xlabel style={font=\color{white!15!black}},
xlabel={Transform},
xtick={1, 4, 7, 10},
ymin=0,
ymax=3.5,
ylabel style={font=\color{white!15!black}},
ylabel={Area},
axis background/.style={fill=white},
axis x line*=bottom,
axis y line*=left,
xmajorgrids,
ymajorgrids
]
\addplot [color=red, line width=2.0pt, forget plot]
  table[row sep=crcr]{%
1	3.14145\\
2	2.2919\\
3	1.96235\\
4	1.8537\\
5	1.7836\\
6	1.74675\\
7	1.71825\\
8	1.7154\\
9	1.7139\\
10	1.7134\\
};
\addplot [color=blue, dashed, line width=2.0pt, forget plot]
  table[row sep=crcr]{%
1	0.677682669173582\\
10	0.677682669173582\\
};
\end{axis}
\end{tikzpicture}%
        }
        \caption{ }
        \label{fig2c}
    \end{subfigure}
    \caption{  
    Example \ref{Example2}: Approximating   $R(1;\, \X_0)$ by applying Proposition \ref{prop:p1} and Theorem \ref{thrm1}. 
    (a) Comparing Proposition \ref{prop:p1} to Theorem \ref{thrm1}: $\X_0$ is shown in red.  $R(1;\, \X_0)$ is shown in green, with a parallelogram over-approximation shown in light green.  $\X_1$ is shown in pink, and a rectangular over-approximation of $R(1;\, \X_1)$ shown in blue.
    (b) Increasing fidelity with multiple transformations: $\X_0$ is shown in red.  $R(1;\, \X_0)$ is shown in green.  An approximation of $R(1;\, \X_0)$ is formed by computing the interaction of 10 approximations derived via Theorem \ref{thrm1}.  This approximation is shown in light green. 
    (c) The red line depicts the surface area of the reachable set approximation as a function of the number of transformation used. After all 10 applications of Theorem \ref{thrm1}, the surface area of the resulting approximation is $1.71$.  The blue line depicts the surface area of the true reachable set $R(1; \, \X_0)$, which is equal to $0.67.$
    }
    \label{fig2}
\end{figure*}

\begin{remark}
When the set of interest $\X_0$ is a singleton set, \emph{i.e.} $\X_0 = [x, x]$, then any nonsingular $T \in \R^{n\times n}$ can be used with Theorem \ref{thrm1} without type-(c) conservatism.  This is due to the fact that $\X_0 = [T^{-1}x, T^{-1}x]_T$ as well.
\end{remark}

\section{Reducing Conservatism via the Use of Multiple Partial Orders}\label{Sec7}
The main focus of this paper is to discover means of improving fidelity in the approximations of reachable sets for nonlinear systems using the theory of  mixed-monotonicity, and we have shown previously how type-(a) conservatism can be mitigated using an approach based on analyzing decomposition functions. In this section, we take a different approach, and show how multiple partial orders, as in those discussed in Section \ref{Sec6}, can be employed to reduce types-(b), (c) and (d) conservatism.

We first turn our attention to type-(d) conservatism.  As suggested in the previous discussion, a naive approach for deriving tighter approximations of reachable sets is to construct several decomposition functions for \eqref{sys} and then form an approximation of the reachable set of \eqref{sys} as the intersection of the approximations derived from each decomposition function and Proposition \ref{prop:p1}.  This approach is, however, unnecessarily complicated since, by Proposition \ref{prop:p2}, multiple decomposition functions for \eqref{sys} can be combined to form a decomposition function that achieves  approximations of reachable sets at least as tight as that computed via this intersection-based approach. Moreover, the application of \eqref{comb_decomp} is still subject to type-(d) conservatism, as this approximation conservatism is inherent in Proposition \ref{prop:p1}.

Nonetheless, we show in Example \ref{Example2} how type-(d) conservatism is mitigated by applying the results of Section \ref{Sec6}. In particular, we show how a decomposition for \eqref{sys} and a decomposition function for \eqref{trans_sys} are used together to approximate reachable sets with added fidelity.

\begin{example}\label{Example2}
We consider the system \eqref{EX1_sys}, previously studied in Example \ref{Example1}.  We take $T$ and $\X_0 = [\underline{y},\, \overline{y}]_T$ from \eqref{EX1_params} and we aim to approximate $R(1,\, \X_0)$ by applying Theorem \ref{thrm1}.  

As in Example \ref{Example1}, we assume access to the tight decomposition function for \eqref{trans_sys} and $\X_0$ is parallelotopic; thus, Theorem \ref{thrm1} can be employed without types-(a) and (c) conservatism in the approximation. Additionally, types-(b) and (d) conservatism cannot be mitigated by, \emph{e.g.}, computing an alternative decomposition function for \eqref{trans_sys} and then forming an approximation of $R(1,\, \X_0)$ as the intersection of the approximations derived from each decomposition function.  Nonetheless, it is possible to reduce overall conservatism by applying Theorem \ref{thrm1} several times with different transformations, so that $R(1,\, \X_0)$ is constrained to the intersection of each approximation derived.

To demonstrate this assertion, we take $\X_1 = [0,\, 3/4] \times [-1/4,\, 1/4]$, so that $\X_0 \subset \X_1$, and we compute a rectangular over-approximation of $R(1;\, \X_0)$ by applying Proposition \ref{prop:p1} with a decomposition function for \eqref{EX1_sys}. We show $\X_1$ and a rectangular approximation of $R(1;\, \X_0)$ graphically in Figure \ref{fig2a}. Note that applying Proposition \ref{prop:p1}, in this case, is subject to type-(b) conservatism as $\X_0$ is not hyperrectangular and we find that the initial application of Theorem \ref{thrm1} leads to a significantly tighter approximation of $R(1;\, \X_0)$ than is attainable using Proposition \ref{prop:p1}. Nonetheless, fidelity is best improved when both Proposition \ref{prop:p1} and Theorem \ref{thrm1} are employed, so that $R(1,\, \X_0)$ is constrained to lie in the intersection of the approximations derived from each.

To illustrate this point further, we next form an approximation of $R(1;\, \X_0)$ by applying Theorem \ref{thrm1} with 10 different transformations matrices; an approximation of $R(1;\, \X_0)$ is then formed as as the intersection of the approximation derived from each transformation (See Figures \ref{fig2b}--\ref{fig2c}), this approach yields a significantly tighter approximation of $R(1;\, \X_0)$ than the initial application of Theorem \ref{thrm1}.
\end{example}

As demonstrated in Example \ref{Example2}, overall conservatism can be reduced when multiple approximations are derived from the application of Theorem \ref{thrm1} with different partial orders. Each approximation, on its own, contains types-(a), (b), (c) and (d) conservatism, however ultimately fidelity is improved in the approach.

It is important to note that, in certain instances, a transformation $T$ can be chosen so that \eqref{trans_sys} is a monotone system as defined in Definition \ref{def:monotone} and, in this instance, the application of Theorem \ref{thrm1} is devoid of type-(b) conservatism.  In this case a parallelotope approximation of $R(t;\, \X_0)$ can be derived such that no proper parallelotope subset of this approximation contains $R(t;\, \X_0)$.  A demonstration is shown in Example \ref{Example3}.

\begin{example}\label{Example3}
Consider the system
\begin{equation}\label{EX3_sys}
\begin{bmatrix}
\dot{x}_1\\
\dot{x}_2
\end{bmatrix}
=
F(x,\, w)
=
\begin{bmatrix}
x_1 - x_2 + x_2^3 + w\\
x_1 - x_2
\end{bmatrix}
\end{equation}
with state-space $\X = \R^2$ and disturbance-space $\W = [-1,\, 1]$. Under the transformation 
\begin{equation}
    T_1 = \begin{bmatrix}
    1 & 1 \\ 0 & 1
    \end{bmatrix}
\end{equation}
the dynamics of $y$ from \eqref{trans_sys} become 
\begin{equation}\label{monotone_ex}
    \begin{bmatrix}
    \dot{y}_1 \\ \dot{y}_2
    \end{bmatrix}
    = F_{T_1}(y,\, w) = 
    \begin{bmatrix}
    y_2^3 + w \\ y_1
    \end{bmatrix},
\end{equation}
and \eqref{monotone_ex} is a monotone system as defined in Definition \ref{def:monotone}. Thus, the application of Theorem \ref{thrm1} with $T_1$ will not be subject to type-(b) conservatism.  An example is shown in Figure \ref{fig3}, where we compute an over-approximation $R(1;\, x_0)$ with $x_0 = (1,\, 1)$ by applying Theorem \ref{thrm1} with a tight decomposition function for \eqref{trans_sys} with $T_1$.  Note that, in this case, the approximation is not subject to types-(a), (b) and (c) conservatism as a tight decomposition function is employed, $x_0$ is trivially a parallelotope, and \eqref{trans_sys} is monotone.  Nonetheless, the approximation derived from the application of Theorem \ref{thrm1} still contains some conservatism, and this is a result of type-(d) conservatism in the method.

Even though \eqref{EX3_sys} is transformable to a monotone system via $T_1$, fidelity in the approximation can still be improved by applying Theorem \ref{thrm1} again with a different shape matrix. An example is shown in Figure \ref{fig3}, where we compare the approximation derived with $T_1$ to a second approximation derived using 
\begin{equation}
    T_2 = \begin{bmatrix}
    1 & 4 \\ -1 & 1
    \end{bmatrix}.
\end{equation} 
While the approximation derived from $T_2$ hugs the boundary of $R(1;\, x_0)$ less tightly than the approximation derived from $T_1$---this is a result of type-(b) conservatism---we find that employing Theorem \ref{thrm1} with both $T_1$ and $T_2$ yields a the tighter approximation of $R(1;\, x_0)$ than was attainable previously.  Moreover, this approach mitigates type-(d) conservatism in the method. \qedhere
\end{example}

\begin{figure}[t]
    \hspace{.6cm}
    \input{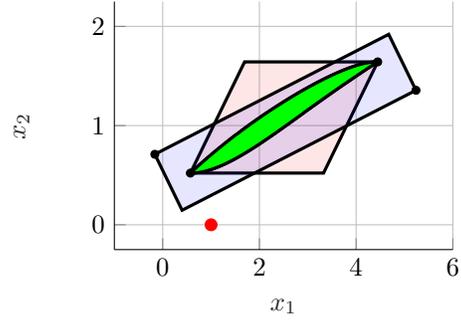}
    \caption{Example \ref{Example3}. $x_0$ is shown in red and $R(1;\, x_0)$ is shown in green. Two parallelogram approximation of $R(1;\, x_0)$ are computed using Theorem \ref{thrm1} with $T_1$ and $T_2$ and these approximations are shown in pink and blue, respectively.  Note that while $T_1$ induces a monotone system, and achieves the tightest parallelogram containing $R(1;\, x_0)$, conservatism is still reduced when repeating the procedure with $T_2$.
    }
    \label{fig3}
\end{figure}

Examples \ref{Example2} and \ref{Example3} demonstrate a new approach for increasing fidelity in the approximations of reachable sets for nonlinear systems using the theory of  mixed-monotonicity.  This approach involves computing a polytope approximation of $R(t;\, \X_0)$ as the \emph{intersection} of several parallelotope approximations derived via Theorem \ref{thrm1}, and we address types-(a), (b), (c) and (d) conservatism individually in discussion. In the next section, we present a numerical example and demonstrate a novel method for reducing type-(c) conservatism when the initial set $\X_0$ is polytopic, and this method forms an approximation of $R(t;\, \X_0)$ \emph{union} of several approximations derived via Theorem \ref{thrm1}.

\section{Numerical Example}\label{Sec8}

Consider the system 
\begin{equation}
    \begin{bmatrix}
    \dot{x}_1 \\
    \dot{x}_2
    \end{bmatrix} = F(x,\, w) = 
    \begin{bmatrix}
    x_2+\sin{(x_2)} + w \\
    x_1 + \cos{(x_1)} + 1
    \end{bmatrix}
\end{equation}
with state-space $\X = \R^2$ and disturbance space $\W = [0,\, 1/2]$, and consider a hexagonal set of initial conditions 
\begin{multline}
    \X_0 = \textbf{Conv}(\{x\in \R^2 \,\vert\, x_1 = 1+ \cos{(\tfrac{i\pi}{3})}, \\
    x_2 = 1+\sin{(\tfrac{i\pi}{3})}, i \in\{ 1,\, \cdots,\, 6\} \})
\end{multline}
where $\textbf{Conv}$ denotes the convex hull function.  We aim to overapproximate $R(1;\, \X_0)$.

One approach for approximating $R(1;\, \X_0)$ is to apply Theorem \ref{thrm1} with a parallelogram  over-approximation of $\X_0$; but this approach---which is taken in Example \ref{Example2}---is subject to type-(c) conservatism. In this study, we instead take an approach whereby $\X_0$ is described exactly as the union of three parallelotopes. An approximation of $R(1;\, \X_0)$ is then computed as the union of three approximations derived via Theorem \ref{thrm1}.

We first describe $\X_0$ as the union of three \emph{disjoint} parallelograms, $\X_0^i = [\underline{y}_i,\, \overline{y}_i]_{T_i}$, for $i\in \{1,\, 2,\, 3\}$, and
\begin{equation*}
    \underline{y}_i = (-1,\, 0) + T_i^{-1} (1,\, 1), \qquad
    \overline{y}_i = (0,\, 1) + T_i^{-1} (1,\, 1)
\end{equation*}
\begin{equation}
    T_i = \begin{bmatrix}
    -\cos{(\tfrac{2 \pi (i-1)}{3})} & \cos{(\tfrac{2 \pi i}{3})} \\
    -\sin{(\tfrac{2 \pi (i-1)}{3})} & \sin{(\tfrac{2 \pi i}{3})}
    \end{bmatrix}.
\end{equation}
Thus, $\X_0 = \cup_{i=1}^3 \X_0^i$, and these parallelograms are disjoint in the sense that $\X_0^1$, $\X_0^2$ and $\X_0^3$ share no common interior points. For each shape matrix $T_i$, a tight decomposition function is formed for \eqref{trans_sys} and the time-1 reachable set of \eqref{trans_sys} is approximated using Theorem \ref{thrm1}.  Then, an approximation $R(1;\, \X_0)$ is formed as the union of the three approximations derived for $R(1;\, \X_0^i)$ with $i\in {1,\, 2,\, 3}$.  We show $\X_0^1$, $\X_0^2$ and $\X_0^3$ graphically in Figure \ref{fig4a} along with their respective reachable set approximations derived in this study. Note that the three approximations derived do not share many common points, and this is a result of the fact that $\X_0^1$, $\X_0^2$ and $\X_0^3$ are chosen to be disjoint.

We next repeat the procedure, and describe $\X_0$ as the union of three \emph{overlapping} parallelograms, as shown in Figure \ref{fig4b}. As was the case previously, a tight decomposition function is formed for each transformed system \eqref{trans_sys} that arises from the shape matrices of these initial parallelograms. For each, the time-1 reachable set of \eqref{trans_sys} is approximated using Theorem \ref{thrm1}, and then $R(1;\, \X_0)$ is approximated as the union of the three approximations derived. Note that the three approximations derived here overlap significantly, and this is a result of the fact that the chosen initial sets overlap. 

\begin{figure}[t]
    \begin{subfigure}{0.48\textwidth}
         \hspace{.6cm}
        \input{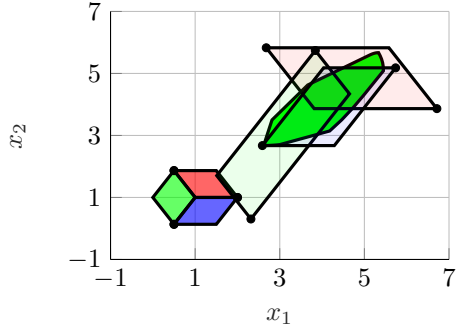}
        \caption{The initial set $\X_0$ is the union of three disjoint parallelograms $\X_0^1$, $\X_0^2$ and $\X_0^3$, shown red, green and blue. The reachable set of each is approximated using theorem \ref{thrm1} and is shown in respective colors.  The true reachable set $R(1;\, \X_0)$ is shown in green.
         \vspace{.5cm}
        }
        \label{fig4a}
    \end{subfigure}
    ~
    \begin{subfigure}{0.48\textwidth}
        \hspace{.6cm}
        \input{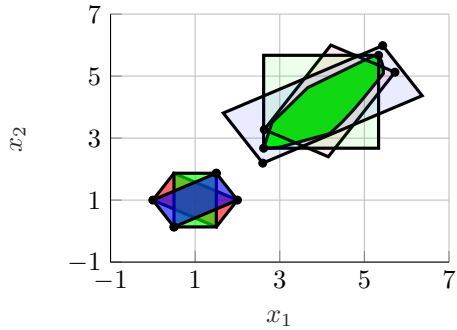}
        \caption{The initial set $\X_0$ is the union of three overlapping parallelograms, shown red, green and blue. The reachable set of each is approximated using theorem \ref{thrm1} and is shown in respective colors.  The true reachable set $R(1;\, \X_0)$ is shown in green.}
		\label{fig4b}
    \end{subfigure}
    \caption{ Numerical Example. Approximating   $R(1;\, \X_0)$ where $\X_0$ is the union of parallelograms. 
    }
    \label{fig4}
\end{figure}

In summary, in this study, we show how reachable sets for initial sets that are not hyperrectangles are approximated using Theorem \ref{thrm1}, and this procedure avoids type-(c) conservatism in the approach.  Note that this procedure is applicable to all systems \eqref{sys} and all polytope initial sets $\X_0$ with hyperrectangular faces.

\section{Conclusion}
This work studies means of improving fidelity in the approximations of reachable sets for nonlinear systems using the theory of  mixed-monotonicity. Four main forms of conservatism are considered, and we show how applying the tools of mixed-monotonicity to a related system, formed via a linear transformation of the initial state-space, is used to reduce this conservatism.

\bibliography{Bibliography}
\bibliographystyle{ieeetr}
\end{document}